\documentclass[a4paper,preprintnumbers,floatfix,superscriptaddress,aps,10pt,twocolumn,notitlepage,longbibliography,noarxiv]{revtex4-1}
\usepackage[utf8]{inputenc}
\usepackage[english]{babel}
\usepackage[T1]{fontenc}
\usepackage{amsmath}
\usepackage{amssymb}
\usepackage{amsthm}
\usepackage{mathtools}
\usepackage[colorlinks=true,citecolor=blue,urlcolor=magenta,linkcolor=red]{hyperref}
\usepackage{todonotes}
\setuptodonotes{inline}
\usepackage{scalerel}
\usepackage{braket}
\usepackage{acronym}

\usepackage[caption=false]{subfig}
\usepackage{accents}
\usepackage{bm,bbm}
\usepackage{lipsum}
\usepackage{nameref}
\usepackage[capitalize]{cleveref}
\crefname{methods}{Methods}{Methods}
\usepackage{array} 
\usepackage{mdframed}
\usepackage{comment}
\usepackage{bm}
\usepackage{listings}
\usepackage{capt-of} 
\usepackage{graphicx}
\usepackage{xcolor}
\usepackage{booktabs}
\usepackage{tabularx}
\usepackage{multirow}

\usepackage{tikz}
\usepackage[outline]{contour}

\newcommand{\RR}{\mathbb{R}}

\renewcommand{\i}{\ensuremath\mathrm{i}} 
\renewcommand{\Pr}{\mathbb{P}} 

\newcommand{\veps}{\varepsilon}
\DeclareMathOperator{\asin}{asin}

\let\Set\undefined

\DeclarePairedDelimiterX\Set[1]\{\}{%
  
  #1
}
\DeclarePairedDelimiterX{\abs}[1]{\lvert}{\rvert}{%
  \ifblank{#1}{\,\cdot\,}{#1}
} 

\DeclarePairedDelimiterX{\mean}[1]{\langle}{\rangle}{%
  \ifblank{#1}{\,\cdot\,}{#1}
} 

\newcommand{\sax}[1]{\sqrt{1-\mean{a_{#1}}^2}}

\newcommand{\xor}{{x_0\oplus x_1}}

\makeatletter
\newtheorem*{rep@theorem}{\rep@title}
\newcommand{\newreptheorem}[2]{%
\newenvironment{rep#1}[1]{%
 \def\rep@title{#2 \ref{##1}}%
 \begin{rep@theorem}}%
 {\end{rep@theorem}}}
\makeatother
\makeatletter
\newtheorem*{rep@lemma}{\rep@title}
\newcommand{\newreplemma}[2]{%
\newenvironment{rep#1}[1]{%
 \def\rep@title{#2 \ref{##1}}%
 \begin{rep@lemma}}%
 {\end{rep@lemma}}}
\makeatother
\usepackage{thmtools}
\usepackage{thm-restate}
\newtheorem{theorem}{Theorem}
\newreptheorem{theorem}{Theorem}

\newenvironment{proofsketch}{%
\proof}{\endproof}
  
\crefname{figure}{Fig.}{Figs.}
\crefname{algorithm}{Protocol}{Protocols}
\crefname{definition}{Def.}{Defs.}

\definecolor{mari}{rgb}{1.0, 0.0, 0.22}

\begin{document}
\title{Information Causality Characterizes the Set of Quantum Correlations\\ in the Simplest Bell Scenario}

\author{Mariami Gachechiladze}
\thanks{\{nikolai.miklin,mariami.gachechiladze\}@tu-darmstadt.de}
\affiliation{Department of Computer Science, Technical University of Darmstadt, Darmstadt, Germany}
\author{Nikolai Miklin}
\thanks{\{nikolai.miklin,mariami.gachechiladze\}@tu-darmstadt.de}
\affiliation{Institute for Applied Physics, Technical University of Darmstadt, Darmstadt, Germany}

\begin{abstract}
Information causality (IC) was introduced as a physical principle
constraining correlations in non-signaling theories. Whether it
can recover the exact quantum correlation boundary, beyond Uffink's inequality, has remained an open question. Here, we combine its generalized formulation for correlated inputs with a new communication protocol to derive quantum Bell inequalities that exactly characterize the quantum correlations in the simplest bipartite Bell scenario, with two binary measurements. In particular, we derive the
Tsirelson-Landau-Masanes criterion directly from IC. Thus, the generalized IC implies macroscopic locality, while we also present macroscopically local correlations that violate
generalized IC. This establishes that generalized IC is a strictly
stronger principle in this scenario. Together with our earlier
result that generalized IC implies a nontrivial communication
complexity principle, these findings strengthen the role of information causality in explaining the limits of quantum nonlocality and provide a systematic route to deriving tighter bounds on the set of quantum correlations in more general Bell scenarios.
\end{abstract}
\maketitle

Quantum theory permits violations of Bell inequalities~\cite{bell1964einstein, brunner2014bell}, but excludes correlations compatible with no-signaling alone. Explaining this gap through physical principles, without invoking the Hilbert-space formalism, is a central goal of quantum foundations. Information causality (IC) approaches this problem by limiting the information a receiver can access about a sender's data to the amount communicated, even in the presence of shared nonlocal resources~\cite{pawlowski2009information}. Its original application recovered Tsirelson's bound~\cite{cirelson1980quantum} on the Clauser–Horne–Shimony–Holt (CHSH) inequality~\cite{clauser1969proposed} and Uffink's stronger quadratic constraint~\cite{uffink2002quadratic,gachechiladze2022quantum}, establishing a direct connection between restrictions on information transfer and limits on nonlocality.

Other principles constrain correlations through computational or macroscopic requirements. Non-trivial communication complexity (NTCC) rules out theories in which every distributed Boolean function can be computed with constant communication and a success probability bounded away from one half~\cite{vandam2013implausible,brassard2006limit, kushilevitz1996communication}. Macroscopic locality (ML) requires that coarse-grained measurements on many independent copies admit a local hidden-variable description in the macroscopic limit~\cite{navascues2010glance}. Neither requirement singles out quantum correlations: both are satisfied by the almost-quantum set, which contains probability distributions with no quantum realization~\cite{navascues2015almost, navascues2007bounding}. Nevertheless, for two binary measurements per party, the quantum, macroscopically local and almost-quantum sets have identical projections onto the four bipartite correlators~\cite{navascues2010glance,navascues2015almost}. This region is characterized exactly by the Tsirelson-Landau-Masanes (TLM) criterion~\cite{cirelson1980quantum,landau1988empirical,masanes2003necessary}.

Recovering this boundary from IC has remained an open challenge. The limited reach of conventional IC tests, together with evidence that post-quantum correlations can evade them, has fueled doubts about the principle's ability to distinguish quantum from more general theories~\cite{pollyceno2023information,navascues2015almost}. These limitations, however, need not reflect the strength of the underlying principle. IC can exclude some macroscopically local correlations~\cite{cavalcanti2010macroscopically}, and we recently established that an extended formulation of IC implies NTCC~\cite{miklin2026communication}. Progress has also come from new tools for extracting its consequences. Noisy communication channels provide an alternative to concatenated protocols~\cite{miklin2021information}, while expansions around vanishing channel capacity yield explicit polynomial constraints on nonlocal correlations~\cite{jain2024informationcausality}. Crucially, allowing correlations among the sender's inputs leads to a formulation in terms of conditional mutual information, derived from the same information-theoretic axioms, that strengthens these constraints beyond Uffink's inequality~\cite{jain2024informationcausality}.

Here, we show that information causality exactly recovers the quantum correlator region for two binary measurements per party. Combining its correlated-input formulation with a new communication protocol, we exploit input correlations, the full information available to Bob, and suitably chosen noisy channels to derive the complete TLM criterion. This closes the gap between earlier IC bounds and the exact quantum boundary in this scenario, providing an operational, information-theoretic characterization of all quantum-realizable bipartite correlators. We further show that the agreement with macroscopic locality in the vanishing-capacity limit does not imply equivalence of the two principles: at finite channel capacity, we identify macroscopically local correlations that violate generalized IC. Thus, information causality is strictly stronger than macroscopic locality in the simplest Bell scenario. Beyond these results, our protocol provides a concrete new route for using information causality to derive tight constraints in more general bipartite Bell scenarios.

\subsection*{Information causality principle}
We first recall the standard formulation of the information causality (IC) principle~\cite{pawlowski2009information}. Alice receives two mutually independent bits $x_0,x_1$, while Bob receives an independent index $y\in\{0,1\}$ and produces a guess $g$ for $x_y$. If Alice communicates through a classical channel of capacity $\mathcal{C}$, Information causality requires
\begin{equation}\label{eq:ic_statement}
   I(x_0;g \vert y=0) +I(x_1;g \vert y=1) \leq \mathcal{C}, 
\end{equation}
where $I(\cdot;\cdot )$ is the Shannon mutual information between $x_i$ and Bob's guess $g$ of Alice's bits $x_0$, $x_1$.

We build on the generalization of IC to correlated inputs introduced in our earlier work~\cite{jain2024informationcausality}, with a further refinement that will be central to our derivation. Rather than quantifying Bob's information through a particular guess $g$, we retain his full classical transcript $(m',b)$, where $m'$ is the message received through the noisy channel and $b$ is his output from the shared nonlocal resource. This distinction matters because converting the transcript into a guess can discard information. For an arbitrary joint distribution of Alice's input bits (which will have another key role in deriving TLM inequality), we therefore consider the bound
\begin{equation}\label{eq:ic_correlated}
  \hspace{-0.15cm} I(x_0;m',b\vert y=0)
    +I(x_1;m',b\vert y=1,x_0)
    \leq I(m;m'),
\end{equation}
where $m$ is Alice's message before transmission and $y\in\{0,1\}$ is Bob's independently chosen input.

The conditioning on $x_0$ in the second term quantifies information about $x_1$ beyond that already supplied by $x_0$, accounting for correlations between Alice's inputs. Retaining the transcript serves a different purpose: since any guess is obtained by processing $(m',b)$ for a given $y$, the data-processing inequality guarantees that replacing the guess with the transcript can only increase the left-hand side. The right-hand side is the mutual information transmitted through the noisy channel for the actual message distribution, and satisfies $I(m;m')\leq\mathcal{C}$.
We refer to~\cref{fig:scenario} for a schematic description of the scenario.  

\begin{figure}
    \centering
    \includegraphics[width=\linewidth]{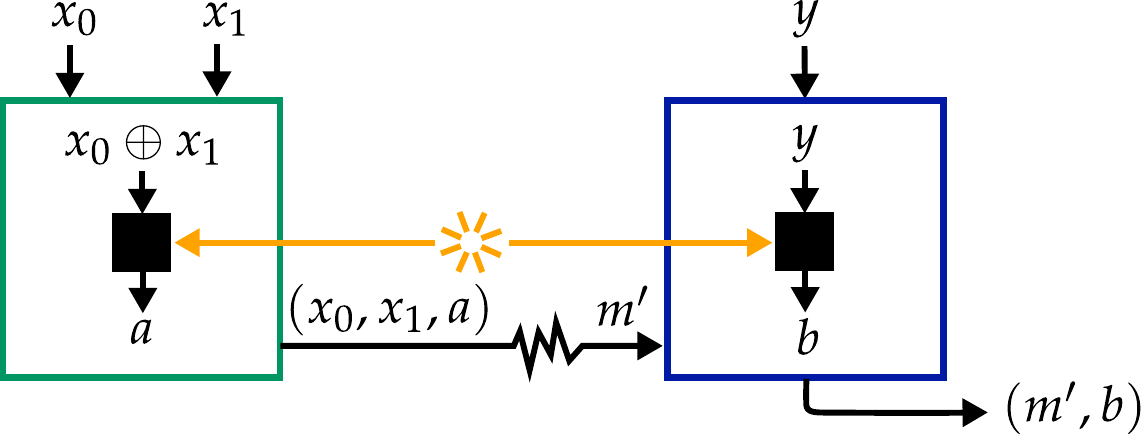}
    \caption{The scenario of information causality. Alice receives two random bits $x_0$, $x_1$, which may be correlated. Bob receives a random bit $y$ independent of $x_0$ and $x_1$. In this work, the parties communicate through a noisy channel that takes the tuple $(x_0,x_1,a)$ as input and produces a bit outcome $m'$. The parties assist their communication with a shared non-signaling resource with binary inputs and outputs.}
    \label{fig:scenario}
\end{figure}

\subsection*{Tsirelson-Landau-Masanes criterion}
In the simplest Bell scenario, i.e., the Bell scenario where Alice and Bob have binary choices of their inputs, $x$, and $y$, and binary outcomes, $a$ and $b$, the correlations that they can establish in quantum theory are known to satisfy the TLM criterion,
\begin{equation}\label{eq:TLM}
    \abs{\asin D_{0,0}+\asin D_{0,1}+\asin D_{1,0}-\asin D_{1,1}}\leq \pi,
\end{equation}
where re-normalized two-body correlators are given by
\begin{equation}\label{eq:Dxy}
    D_{x,y} \coloneqq \frac{\mean{a_xb_y}-\mean{a_x}\mean{b_y}}{\sqrt{(1-\mean{a_x}^2)(1-\mean{b_y}^2)}},
\end{equation}
with $\mean{a_xb_y}\coloneqq\sum_{a,b\in\Set{0,1}}(-1)^{a\oplus b}\Pr(a,b\vert x,y)$, $\mean{a_x}\coloneqq \sum_{a\in\Set{0,1}}(-1)^a\Pr(a\vert x)$, and $\mean{b_y}\coloneqq \sum_{b\in\Set{0,1}}(-1)^b\Pr(b\vert y)$.
The inequality in \cref{eq:TLM} has been independently derived by Tsirelson~\cite{cirelson1980quantum}, Landau~\cite{landau1988empirical}, and Masanes~\cite{masanes2003necessary}. Most notably, whenever $\mean{a_x}=\mean{b_y}=0$, TLM inequality is known to characterize the set of quantum correlations, i.e., provide a necessary and sufficient condition for the tuple $(\mean{a_0b_0},\mean{a_0b_1},\mean{a_1b_0},\mean{a_1b_1})$ to have a quantum realization.
Navascu\'es, Pironio, and Ac\'in have also shown~\cite{navascues2007bounding} that \cref{eq:TLM} is equivalent to the condition that the correlators, including the marginals $\mean{a_x}$, $\mean{b_y}$, satisfy the principle of ML.

\section*{Results}

\begin{theorem}\label{th}
    Generalized information causality characterizes the set of quantum correlations in the simplest bipartite Bell scenario with two binary inputs and outputs. Specifically, it implies the Tsirelson-Landau-Masanes criterion.
\end{theorem}

Here, we give a proof sketch by presenting the protocol and outlining the main steps. Technical details are provided in~\cref{app:theorem1}. 

\begin{proofsketch}
We start by introducing the protocol. Alice's and Bob's measurement choices are $x=\xor$ and $y$, where $\oplus$ is a modulo $2$ addition. Alice sends a message $m=(x_0, x_1, a)$ through a noisy channel which maps an eight-dimensional input to a binary output $m'$. The conditional mutual information is then measured between Alice's input bits $x_b$ and the final output tuple on Bob's side $(m', b)$. Alice's input bit correlations are specified by
\begin{equation}
    \Pr(x_0, x_1)=\frac{1+(-1)^\xor\veps}{4}, 
\end{equation}
where $\veps\in[-1, 1]$. This keeps both bits individually uniform while allowing their parity probabilities to vary. The exact choice of $\veps$ and the classical channel parameters have an important role in the proof and are described below. 

Our derivation rests on three changes to the standard protocol. First, Alice supplies her complete classical record $m=(x_0,x_1,a)$ as the input to the noisy channel, while Bob's information is evaluated on the joint transcript $(m',b)$, rather than on a particular guess. Bob receives only the binary channel output $m'$; retaining it together with his outcome $b$ avoids the information loss associated with a decoding rule. Crucially, this preserves the dependence on Bob's marginal distribution: in the quadratic expansion of the information bound, the contributions associated with each outcome are weighted by $1/\Pr(b\vert y)$. Bob's local statistics therefore enter the resulting constraint explicitly.

Second, we exploit the freedom to choose the communication channel. Since IC must hold for every classical channel, it must also hold for channels whose transition probabilities are chosen as functions of Alice's local marginals. Specifically, we take 
\begin{equation}
    \Pr(m'\vert m) =
    \frac12+\frac{(-1)^{m'}e_c}{2}f_m, \quad \abs{e_c f_m}\leq 1.
\end{equation}
In this work, we choose
\begin{equation}
   f_{x_0,x_1,a} = (-1)^{a\oplus x_0}
    \frac{\kappa_{\xor}}{\sax{\xor}},
\end{equation} 
where $\kappa_0,\kappa_1$ are arbitrary real parameters. 
For nonzero local marginals and any fixed choice of these parameters, the channel is well defined for sufficiently small $\abs{e_c}$: its validity requires only $\abs{e_c\kappa_x}\leq\sax{x}$ for each $x$. Thus, the dependence on Alice's marginals can be incorporated without restricting the channel parameters that enter the limiting inequality. 

The role of the channel normalization becomes apparent as its capacity approaches zero. We consider channels with arbitrarily small capacity, taking $e_c\to 0$ while keeping
$\kappa_0,\kappa_1$ fixed. Although the mutual-information terms
vanish in this limit, their leading-order coefficients retain
nontrivial information about the shared correlations. Since IC
must hold for every such channel, it also constrains these
coefficients. Dividing the left-hand side of \cref{eq:ic_correlated} by the right-hand side and finding the limit by applying l'H\^{o}pital's rule twice yields a quadratic constraint,
following the approach of Ref.~\cite{jain2024informationcausality}. The exact inequality in terms of channel parameters, shared correlations, and $\veps$ is given in Appendix~\cref{eq:afterlimchi}. 

The channel is chosen to produce the normalization appearing in
$D_{x,y}$. As $e_c\to 0$, the leading contribution to each
mutual-information term is quadratic in the channel parameter $f_m$.
The factor $1/\sax{x}$ built into $f_m$ contributes to the limiting bound. On the other hand, Bob's normalization
arises differently: retaining his outcome $b$ gives separate
contributions weighted by $1/\Pr(b \vert y)$. Summing these
contributions produces the correlators
$\mean{a_xb_y}-\mean{a_x}\mean{b_y}$, normalized by $\sqrt{1-\mean{b_y}^2}$ Bob's local
marginal. Together, the two ingredients yield $D_{x,y}$.

The final ingredient is the choice of correlations between Alice's
input bits. Since IC must hold for every source distribution, we
can choose $\veps$ as a function of the shared correlators.
After exploiting the freedom to vary the channel parameters
$\kappa_0,\kappa_1$, the optimal choice is
\begin{equation}
    \frac{1+\veps}{1-\veps}=\sqrt{\frac{(1-D_{0,1}^2)(1-D_{1,0}^2)}{(1-D_{0,0}^2)(1-D_{1,1}^2)}}.
\end{equation}
With this choice, the derived constraint combines exactly into the TLM form (as given in Ref.~\cite{landau1988empirical}).
\end{proofsketch}

Now that we have established that the generalized information causality can exactly reproduce the boundary of correlations in the macroscopic local set, we show that the correlations compatible with IC are a strict subset of those with ML. 

\begin{theorem}
    The generalized information causality principle is strictly stronger than Macroscopic Locality in the simplest Bell scenario. 
\end{theorem}

While it is known that macroscopic local correlations can violate IC in larger Bell tests~\cite{cavalcanti2010macroscopically}, inclusions of these two sets of correlations have not been established in any scenario. On the contrary, before this work, it was believed that Information causality is a weaker physical principle~\cite{pollyceno2023information, jain2026nonlocal, jain2026bounds}.

\begin{proof}
To prove this statement, it is sufficient to present bipartite correlations that satisfy the ML principle but violate the IC.  

We express such probabilities in the table, 
\begin{equation}
\label{eq:MLbox}
\renewcommand{\arraystretch}{1.3}
\setlength{\arraycolsep}{8pt}
\Pr(a,b|x,y)=
\begin{array}{c|cccc}
    xy\backslash ab
    & 00 & 01 & 10 & 11 \\
    \hline
    00 & \tfrac34 & \tfrac1{20} & 0 & \tfrac15 \\
    01 & \tfrac12 & \tfrac3{10} & 0 & \tfrac15 \\
    10 & \tfrac34 & \tfrac1{20} & 0 & \tfrac15 \\
    11 & \tfrac3{10} & \tfrac12 & \tfrac15 & 0
\end{array}\ .
\end{equation}

One can directly obtain single and two body correlations $\mean{a_0}=\mean{a_1}=\frac{3}{5}$, $\mean{b_0}=\frac{1}{2}$ and $\mean{b_1}=0$, while $\mean{a_0b_0}=\mean{a_1b_0}=\frac{9}{10}$ and $\mean{a_0b_1}=-\mean{a_1b_1}=\frac{2}{5}$. Inserting them in~\cref{eq:Dxy} gives us four re-normalized two-body correlators, 
\begin{equation}
    \left(\frac{\sqrt{3}}{2},\frac{\sqrt{3}}{2},\frac{1}{2},-\frac{1}{2}\right), 
\end{equation}
which clearly give the highest value of $\pi$ in~\cref{eq:TLM} and, thus, satisfy the TLM criterion.

Next, we show that these macroscopic local correlations violate the generalized IC principle. Note that for this we choose a new protocol and a nonzero channel capacity, which is, in general, known to give stronger constraints on quantum correlations in some cases~\cite{miklin2021information}. 

For this example, consider the uniformly distributed random inputs of Alice, $x_0, x_1\in\{0,1\}$ (i.e., $\veps=0$), and the van Dam protocol~\cite{van2013implausible}, in which Alice's choice of measurement is $x=\xor$, and the communication channel is binary and perfect $m'=m=a\oplus x_0$. This results in the right-hand side of the IC statement in \cref{eq:ic_correlated} being equal to $1$, if we choose the logarithm to be base $2$.

Following the steps of the exact derivations in~\cref{app:theorem1}, we explicitly derive two terms of the mutual information on the left-hand side of the IC inequality. Note that even though we consider the standard setting of IC with uncorrelated input bits, the expression in \cref{eq:ic_correlated} does not reduce to \cref{eq:ic_statement}, because we do not introduce the guess variable $g$.

First, we evaluate the conditional distribution
\begin{align}\label{eq:Prob_mb_ML}
  \Pr(m,b\vert x_0, x_1, y)& =\sum_a \Pr(m\vert x_0,x_1,a)\Pr(a,b\vert \xor, y)\nonumber \\
    & =\Pr(m \oplus x_0,b\vert \xor, y),
\end{align}
where we used the fact that $m=x_0\oplus a$.
These distributions correspond to the ones in \cref{eq:MLbox}, up to some permutation which does not affect the calculation of the mutual information.

Inserting \cref{eq:Prob_mb_ML} into \cref{eq:ic_correlated},  direct calculation results in the left-hand side of it being
\begin{equation}
    \frac{1}{4}\log 5+\frac{3}{10}\log 3 \approx 1.056>1,
\end{equation}
i.e., a violation of the IC principle.
\end{proof}

\section*{Discussion}
We show that information causality is a strictly stronger physical principle than macroscopic locality in the simplest Bell scenario. We achieve this by  three seemingly modest modifications of the original formulation of the IC principle:
correlated inputs on Alice's side, retention of the received
message jointly with Bob's outcome, and channel parameters
chosen to incorporate Alice's local marginals. For the resulting
protocol, the leading constraint in the limit of vanishing
channel capacity exactly reproduces macroscopic
locality principle in the simplest Bell scenario. For a nonzero capacity, we show that the IC principle leads to tighter results than ML. 
Together with our earlier result that extended IC implies non-trivial communication complexity~\cite{miklin2026communication}, these findings strengthen the position of IC among physical principles constraining bipartite correlations. Whether almost-quantum correlations satisfy generalized IC remains unresolved.

Our results reopen several research directions. The full power
of IC is still unknown, even in the simplest Bell scenario.
Can a different protocol yield constraints beyond ML already
at leading order in the vanishing-capacity limit? Does the
recovery of ML through optimized communication channel extend to
more general bipartite Bell scenarios? Can almost-quantum
correlations violate generalized IC, either in the simplest Bell
scenario or in a larger one? Finally, combining protocols
involving wirings of several boxes with vanishing-capacity
methods may yield stronger, higher-order quantum Bell
inequalities. These questions offer new routes to exploring
how far information causality can account for the limits of
quantum nonlocality.

\acknowledgements
    We thank Costantino Budroni and Christoph Smaczny for inspiring discussions and the Quantum Information 2026 workshop held at the Centre Paul-Langevin in the village of Aussois, France, where the results were derived. 
    This research was funded within the QuantERA II Programme that has received funding from the EU's H2020 research and innovation programme under the GA No 101017733.
    This research work was supported by the National Research Center for Applied Cybersecurity ATHENE.

\bibliography{ref}

\onecolumngrid
\section*{Appendix}
\crefalias{section}{appendix}
\crefalias{subsection}{appendix}
\begin{appendix}
In this Appendix, we provide technical details that support the statements from the main text.

\section{Proof of~\cref{th}}\label{app:theorem1}
We start by writing each mutual information term from the generalized IC statement in \cref{eq:ic_correlated} separately.
\begin{align}
I(x_0;m',b\vert y=0)&=H(m',b\vert y=0)-\sum_{j\in\{0,1\}}\Pr(x_0=j)H(m',b\vert x_0=j, y=0).~\label{eq:ICLHS0}\\
I(x_1;m',b\vert y=1, x_0)&=\sum_{j\in\{0,1\}}\Pr(x_0=j)H(m',b\vert x_0=j, y=1)\nonumber\\&
    -\sum_{j,k\in\{0,1\}}\Pr(x_0=j, x_1=k)H(m',b\vert x_0=j, x_1=k, y=1).~\label{eq:ICLHS1}\\
I(m;m')&=H(m')-\sum_{l\in\{0,1\}^3} \Pr(m=l)H(m'\vert m=l).~\label{eq:ICRHS}
\end{align}
The probabilities of the shared non-signaling resource can be expressed in terms of the correlators as
\begin{equation}\label{eq:pabxy}
    \Pr(a,b\vert x,y) = \frac{1}{4}\left(1+(-1)^a\mean{a_x}+(-1)^b\mean{b_y}+(-1)^{a\oplus b}\mean{a_xb_y}\right),
\end{equation}
for $a,b,x,y\in\Set{0,1}$.
The marginal probabilities are then equal to
\begin{equation}\label{eq:paxpby}
    \Pr(a\vert x) = \frac{1}{2}(1+(-1)^a\mean{a_x}),\quad \Pr(b\vert y) = \frac{1}{2}(1+(-1)^b\mean{b_y}).
\end{equation}
In this work, we consider the communication channel taking the entire data on Alice's side, namely $m=(x_0,x_1,a)$, as its input and producing the output $m'\in\Set{0,1}$.
The channel of this form can be parametrized in the following way
\begin{equation}\label{eq:channel_probs}
    \Pr(m'\vert m) = \frac{1}{2}+(-1)^{m'}\frac{e_c}{2}f_m,
\end{equation}
for $m'\in\Set{0,1}$, $m\in\Set{0,1}^3$, and $e_c,f_m\in \RR$.
As long as $\abs{e_cf_m}\leq 1$, the probabilities in \cref{eq:channel_probs} are non-negative and sum up to $1$ for each $m$, i.e., the corresponding channel exists.
The parameter $e_c$ in \cref{eq:channel_probs} plays the role of determining the strength of the noise in the channel.

Next, we write down each probability necessary for calculating the entropy functions in \cref{eq:ICLHS0,eq:ICLHS1,eq:ICRHS}.
\begin{align}\label{eq:probs_gen}
\Pr(m',b\vert x_0,x_1,y) = \sum_a \Pr(m'\vert x_0,x_1,a)\Pr(a,b\vert \xor,y) = \frac{\Pr(b|y)}{2}
+\frac{(-1)^{m'}e_c}{2}
\underbrace{
\sum_a f_{x_0,x_1,a}\Pr(a,b|\xor,y)
}_{\eqqcolon\,\chi_{x_0,x_1,y}^{b}},
\end{align}
where $\chi_{x_0,x_1,y}^{b}\in \RR$, and we substituted the conditioning on $x_0,x_1$ with $\xor$ in the probability above, because we fix Alice's measurement choice to be $\xor$. 
From \cref{eq:probs_gen}, we can calculate the rest of the probabilities
\begin{align}\begin{split}
\Pr(m',b|y=0)
&= \frac{\Pr(b|y=0)}{2}
+\frac{(-1)^{m'}e_c}{2}
\sum_{x_0,x_1} \left(\sum_{a} f_{x_0,x_1,a}\Pr(a,b|\xor,y=0)\right)\Pr(x_0,x_1)\\
&= \frac{\Pr(b|y=0)}{2}
+\frac{(-1)^{m'}e_c}{2}
\sum_{x_0,x_1} \Pr(x_0,x_1) \chi_{x_0,x_1,0}^b,
\end{split}\end{align}
and
\begin{align}\begin{split}
\Pr(m',b|x_0,y)
&= \frac{\Pr(b|y)}{2}
+\frac{(-1)^{m'}e_c}{2}
\sum_{x_1} \left(\sum_{a} f_{x_0,x_1,a}\Pr(a,b|\xor,y=0)\right)\Pr(x_1\vert x_0)\\
&= \frac{\Pr(b|y)}{2}
+\frac{(-1)^{m'}e_c}{2}
\sum_{x_1} \Pr(x_1|x_0) \chi_{x_0,x_1,0}^b.
\end{split}\end{align}
We now use the technique, which we introduced in Ref.~\cite{jain2024informationcausality}, to obtain quadratic constraints from the IC statement.
In short, we consider the limit of the IC inequality for the channel capacity going to $0$, which corresponds to $e_c\to 0$. This limit can be obtained using l'H{\^o}pital's rule and effectively replaces each summand in the IC statement with the limit of its second derivative for $e_c\to 0$.
We use the following general formula
\begin{equation}
    \lim_{e_c\to 0}\frac{\partial^2}{\partial e_c^2}(\xi+e_c\eta)\log(\xi+e_c\eta) = \frac{\eta^2}{\xi},
\end{equation}
for $\xi,\eta\in \RR$ to avoid identical calculations for each of the entropy terms. 
This results in the following implication of the IC statement in \cref{eq:ic_correlated},
\begin{align}\label{eq:afterlimchi}
\begin{split}
-&\sum_{b} \frac{\left(\sum_{x_0,x_1} \chi_{x_0,x_1, 0}^b \Pr(x_0,x_1)\right)^2}{\Pr(b|y=0)} + \frac{1}{2}\sum_{x_0,b,y} (-1)^y\frac{\left(\sum_{x_1} \chi_{x_0,x_1,y}^b \Pr(x_1|x_0)\right)^2}{\Pr(b|y)}+\sum_{x_0,x_1} \Pr(x_0,x_1)\sum_b \frac{(\chi_{x_0,x_1,1}^b)^2}{\Pr(b\vert y=1)}\\
\leq&-\left(\sum_{x_0,x_a,a} f_{x_0,x_1,a}\Pr(a\vert \xor) \Pr(x_0,x_1) \right)^2+\sum_{x_0,x_a,a} \Pr(a\vert \xor) \Pr(x_0,x_1)f_{x_0,x_1,a}^2.
\end{split}
\end{align}
We now fix the form of the parameters $f_{x_0,x_1,a}$ to be
\begin{equation}\label{eq:fx0x1a}
 f_{x_0,x_1,a} = (-1)^{a\oplus x_0}\frac{\kappa_{\xor}}{\sax{\xor}},
\end{equation}
for $\kappa_0,\kappa_1\in\RR$. Inserting this expression and the parametrization of probabilities in~\cref{eq:pabxy} in the definition of $\chi_{x_0,x_1,y}^b$, we obtain,
\begin{equation}
    \chi_{x_0,x_1,y}^b = \frac{(-1)^{x_0}}{2}\frac{\kappa_{\xor}}{\sax{\xor}}\left(\mean{a_{\xor}}+(-1)^b\mean{a_{\xor}b_y}\right).
\end{equation}
We incorporate this parametrization in the individual summand in~\cref{eq:afterlimchi}.
For the first summand, we show that the numerator is zero. Indeed,
\begin{align}
    \begin{split}
        \sum_{x_0,x_1} \chi_{x_0,x_1, 0}^b \Pr(x_0,x_1)= \sum_{x_0,x_1} \frac{(-1)^{x_0}}{2}\frac{\kappa_{\xor}}{\sax{\xor}}\left(\mean{a_{\xor}}+(-1)^b\mean{a_{\xor}b_0}\right) \frac{1+(-1)^{\xor}\veps}{4}=0.
    \end{split}
\end{align}
Apart from the factor $(-1)^{x_0}$, the rest of the expression depends on $x_0$ and $x_1$ only through their parity $\xor$. For each fixed value of this parity, the two corresponding pairs $(x_0,x_1)$ therefore contribute with equal magnitude and opposite signs, and hence cancel exactly.

Moving to the second summand, since the two terms only differ in Bob's setting $y$, we keep the input unspecified,
\begin{align}\label{eq:second_sum_pre}
    \begin{split}
        \sum_{x_0,b} \frac{\left(\sum_{x_1} \chi_{x_0,x_1, 0}^b \Pr(x_1|x_0)\right)^2}{\Pr(b|y)}=&\sum_{x_0,b} \frac{\left(\sum_{x_1} \frac{\kappa_{\xor}}{\sax\xor}\left(\mean{a_{\xor}}+(-1)^b\mean{a_{\xor}b_y}\right) \frac{1+(-1)^{\xor}\veps}{2}\right)^2}{2 (1+(-1)^b\mean{b_y}) }\\
        =&\sum_{b} \frac{\left(\sum_{x} \frac{\kappa_{x}}{\sax{x}}\left(\mean{a_{x}}+(-1)^b\mean{a_{x}b_y}\right) \frac{1+(-1)^{x}\veps}{2}\right)^2}{1+(-1)^b\mean{b_y} }.
    \end{split}
\end{align}
We can rewrite the terms in the numerator in the re-normalized correlator notation in \cref{eq:Dxy},
\begin{equation}
\mean{a_x}+(-1)^b\mean{a_xb_y} = \left(1+(-1)^b\mean{b_y}\right)\mean{a_x} + (-1)^b\sax{x}\sqrt{1-\mean{b_y}^2} D_{x,y}.
\end{equation}
This leads to a simplified form of the expression in \cref{eq:second_sum_pre},
\begin{align}
&\sum_b\frac{1}{1+(-1)^b\mean{b_y}}\left[\left(1+(-1)^b\mean{b_y}\right)\sum_x\frac{\kappa_x\mean{a_x}}{\sax{x}}\frac{1+(-1)^x\veps}{2} + (-1)^b\sqrt{1-\mean{b_y}^2}\sum_x\kappa_x D_{x,y}\frac{1+(-1)^x\veps}{2}\right]^2
\nonumber\\
&=2\left[\sum_x\frac{\kappa_x\mean{a_x}}{\sax{x}}\frac{1+(-1)^x\veps}{2}\right]^2+2\left[\sum_x\kappa_x D_{x,y}\frac{1+(-1)^x\veps}{2}\right]^2
\label{eq:LHS23_D}.
\end{align}
Now, when we calculate the full expression for the second summand in \cref{eq:afterlimchi}, e.g., take the sum over $y$, we notice that the first term in \cref{eq:LHS23_D} cancels out, because it does not depend on $y$.
Therefore, we intermediately conclude that 
\begin{equation}
    \frac{1}{2}\sum_{x_0,b,y} (-1)^y\frac{\left(\sum_{x_1} \chi_{x_0,x_1,y}^b \Pr(x_1|x_0)\right)^2}{\Pr(b|y)} = \sum_{y}(-1)^y\left[\sum_x\kappa_x D_{x,y}\frac{1+(-1)^x\veps}{2}\right]^2.
\end{equation}
Finally, the last term on the left-hand side of \cref{eq:afterlimchi} becomes
\begin{align}
    \sum_{x_0,x_1} \Pr(x_0,x_1)\sum_b \frac{(\chi_{x_0,x_1,1}^b)^2}{\Pr(b\vert y=1)} & = \sum_{x\in\Set{0,1}}\frac{1+(-1)^x\veps}{4}\frac{\kappa_x^2}{1-\mean{a_x}^2}\sum_{b}\frac{(\mean{a_x}+(-1)^b\mean{a_xb_1})^2}{1+(-1)^b\mean{b_1}}\\
    &=\sum_{x\in\Set{0,1}}\frac{1+(-1)^x\veps}{2}\kappa^2_x \left(D_{x,1}^2+\frac{\mean{a_x}^2}{1-\mean{a_x}^2}\right).\label{eq:LHS4_D}
\end{align}

Continuing with the terms on the right-hand side of \cref{eq:afterlimchi}, we again notice that the first term vanishes
\begin{equation}
    \sum_{x_0,x_1,a}(-1)^{a\oplus x_0}\frac{\kappa_{\xor}}{\sax{\xor}}\frac{1+(-1)^a\mean{a_{\xor}}}{2}\frac{1+(-1)^{\xor}}{4}=0,
\end{equation}
because of the summation over $x_0$.
The second summand on the right-hand side of \cref{eq:afterlimchi} simplifies to 
\begin{equation}
    \sum_{x_0,x_1,a}\frac{1+(-1)^a\mean{a_\xor}}{2}\frac{1+(-1)^\xor\veps}{4}\frac{\kappa_\xor^2}{1-\mean{a_\xor}^2}=\sum_{x\in\Set{0,1}}\frac{1+(-1)^x\veps}{2}\frac{\kappa_x^2}{1-\mean{a_x}^2}.\label{eq:RHS_D}
\end{equation}

Having expressed all the parts of the expression in \cref{eq:afterlimchi} in terms of $\kappa_{0},\kappa_1$, $\veps$ and the observed correlators, we bring the final expressions in \cref{eq:LHS23_D,eq:LHS4_D,eq:RHS_D} and group the terms with respect to $\kappa_0$ and $\kappa_1$
\begin{align}\label{eq:k0k1_D_condition}
\begin{split}
    \kappa_0^2\left(D_{0,0}^2\left(\frac{1+\veps}{2}\right)^2+D_{0,1}^2\frac{1-\veps^2}{4}-\frac{1+\veps}{2}\right) &+\kappa_1^2\left(D_{1,0}^2\left(\frac{1-\veps}{2}\right)^2+D_{1,1}^2\frac{1-\veps^2}{4}-\frac{1-\veps}{2}\right) \\
    &+2\kappa_0\kappa_1\frac{1-\veps^2}{4}\left(D_{0,0}D_{1,0}-D_{0,1}D_{1,1}\right)\leq 0.
    \end{split}
\end{align}
Now, we notice that the condition in \cref{eq:k0k1_D_condition} for arbitrary $k_0,k_1\in \RR$ is equivalent to a requirement that the following matrix is positive semidefinite
\begin{align}\everymath{\displaystyle}\label{eq:PSD}
    \begin{pmatrix}
        \frac{1+\veps}{1-\veps}(1-D_{0,0}^2)+1-D_{0,1}^2 & D_{0,1}D_{1,1}-D_{0,0}D_{1,0} \\
        D_{0,1}D_{1,1}-D_{0,0}D_{1,0} & \frac{1-\veps}{1+\veps}(1-D_{1,0}^2)+1-D_{1,1}^2
    \end{pmatrix}\geq 0,
\end{align}
where we have further simplified the terms in front of $\kappa_0^2$ and $\kappa_1^2$.
The diagonal terms of~\cref{eq:PSD} are both non-negative.
Therefore, \cref{eq:PSD} is equivalent to the determinant positivity condition,
\begin{equation}
    \left(D_{0,0}D_{1,0}-D_{0,1}D_{1,1}\right)^2\leq \left(\frac{1+\veps}{1-\veps}(1-D_{0,0}^2)+1-D_{0,1}^2\right)\left(\frac{1-\veps}{1+\veps}(1-D_{1,0}^2)+1-D_{1,1}^2\right).
\end{equation}
The final step is to choose the parameter $\veps$, such that 
\begin{equation}
    \frac{1+\veps}{1-\veps} = \sqrt{\frac{(1-D_{0,1}^2)(1-D_{1,0}^2)}{(1-D_{0,0}^2)(1-D_{1,1}^2)}},
\end{equation}
which leads to the TLM inequality (in the form given in Ref.~\cite{landau1988empirical})
\begin{equation}
    \abs{D_{0,0}D_{1,0}-D_{0,1}D_{1,1}}\leq \sqrt{1-D_{0,0}^2}\sqrt{1-D_{0,1}^2}+\sqrt{1-D_{1,0}^2}\sqrt{1-D_{1,1}^2}.
\end{equation}

\end{appendix}

\end{document}